\documentclass[12pt, english]{article}

\usepackage{babel}
\usepackage{amsmath}
\usepackage{amssymb}
\usepackage{amsthm}
\usepackage{tikz}
\usepackage{url}

\theoremstyle{plain}
\newtheorem{Thm}{Theorem} 

\newtheorem{Lemma}[Thm]{Lemma}
\newtheorem*{Lemma*}{Lemma}

\theoremstyle{definition}

\newtheorem{Def*}{Definition}

\begin{document}

\title{Faulty picture-hanging improved}
\author{Johan Wästlund}
\date{\today}

\maketitle

\begin{abstract}  
A picture-hanging puzzle is the task of hanging a framed picture with a wire around a set of nails in such a way that it can remain hanging on certain specified sets of nails, but will fall if any more are removed. The classical brain teaser asks us to hang a picture on two nails in such a way that it falls when any one is detached. 

Demaine et al (2012) proved that all reasonable puzzles of this kind are solvable, and that for the $k$-out-of-$n$ problem, the size of a solution can be bounded by a polynomial in $n$. 
We give simplified proofs of these facts, for the latter leading to a reasonable exponent in the polynomial bound.
\end{abstract}

\section{Introduction} \label{S:intro}

A now famous picture-hanging puzzle was posed by A.~Spivak in 1997. In \cite{Spivak} (see also \cite{Sillke}), a ``brain teaser'' asked for an explanation of the following: A picture hangs on two nails, and the wire is wound around the nails in such a way that the picture will fall if any one of them is removed. One solution is shown in Figure~\ref{F:Spivak}. The puzzle has been popularized by several authors and video channels \cite{Parker, StandUp, Tan-Holmes, Winkler}.

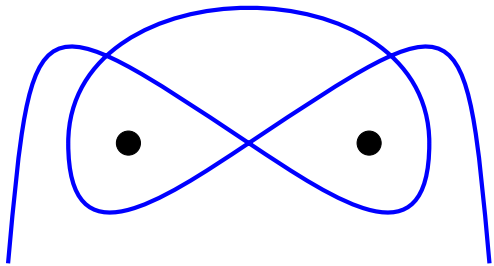
\begin{figure} [h]
\begin{center}
\begin{tikzpicture}[scale=0.8]

\filldraw (-2,0) circle (0.2);
\filldraw (2,0) circle (0.2); 

\draw[blue, ultra thick] (-4,-2) .. controls (-3.6, 2.4) .. (0,0)
(0,0) .. controls ( 1.5, -1) and (3,-2) .. (3,0) 
(3,0) .. controls (3,3) and (-3,3) .. (-3,0)
(-3,0) .. controls (-3,-2) and (-1.5, -1) .. (0,0)
(0,0) .. controls (3.6,2.4) .. (4,-2); 

\end{tikzpicture}
\caption{A solution to Spivak's picture-hanging puzzle. The black dots represent the nails. If any one of them is removed, it becomes apparent that the wire doesn't actually wind around the other, and the picture falls.}
\label{F:Spivak}
\end{center}
\end{figure}

A more detailed account of the history of the puzzle is found in \cite{Demaine}, which also discusses a connection to the Borromean rings. As was observed by Neil Fitzgerald \cite{Sillke}, the solution is a manifestation of the fact that the fundamental group of the plane minus two points is non-abelian. 

\begin{figure} [h]
\begin{center}
\begin{tikzpicture}

\filldraw (-2,0) circle (0.2);
\filldraw (2,0) circle (0.2); 

\draw[blue, ultra thick] (0,-3) .. controls (-5, 0) and (-2,4) .. (0,-3); 
\draw[blue, ultra thick] (0,-3) .. controls (2, 4) and (5,0) .. (0,-3); 

\node at (-3,1.3) {\Large $x$};
\node at (3,1.3) {\Large $y$};

\draw [blue, ultra thick] (-2.5, -1) -- (-2.5, -0.6) -- (-2.1, -0.7);
\draw [blue, ultra thick] (-1.3, -0.5) -- (-0.9, -0.6) -- (-0.9, -0.2);

\draw [blue, ultra thick] (2.7, -0.8) -- (2.3, -0.9) -- (2.25, -0.5);
\draw [blue, ultra thick] (1.1, -0.8) -- (1.0, -0.4) -- (0.6, -0.5);

\end{tikzpicture}
\caption{A pair of generators for the fundamental group of the plane minus two points. Since the group is non-abelian, an expression like $xy^{-1}x^{-1}y$ cannot be simplified to the identity. }
\label{F:fundamental}
\end{center}
\end{figure}
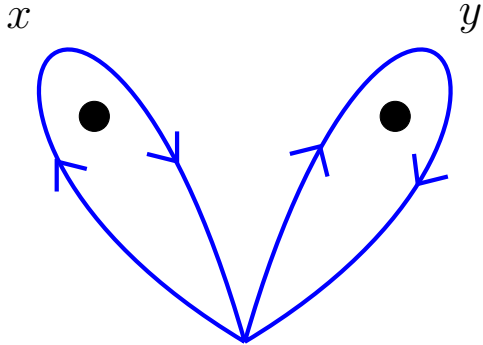

If we choose a base-point, say where the picture is attached to the wire, and take the directed closed curves $x$ and $y$ as generators of the fundamental group as in Figure~\ref{F:fundamental}, then the picture-hanging of Figure~\ref{F:Spivak} can be described by the word \[xy^{-1}x^{-1}y.\] 
Equally valid solutions are obtained by $xyx^{-1}y^{-1}$ etc. Since $x$ and $y$ don't commute, these words represent nontrivial group elements. But if either $x$ or $y$ is set to 1 (``quotiented out''), the whole expression will collapse to the identity. The operation of setting $x$ or $y$ equal to 1 corresponds to removing one of the nails so that the curve around that nail becomes contractible.

Once we have a group-theoretic version of the puzzle, we can prove the obvious generalization to $n$ nails. This was noted by Fitzgerald, Chris Lusby Taylor and others. Again we refer to \cite{Demaine} for the background.

\begin{Thm} \label{T:obvious}
For every $n$, it's possible to hang a picture on $n$ nails so that it falls on the removal of any one of them.
\end{Thm}

\begin{proof}
For $n$ nails let the generators of the fundamental group be $x_1,\dots ,x_n$, where $x_i$ is given by looping clockwise around the $i$:th nail only. Invoking induction, let the word $A$ be a solution for $k$ nails $x_1,\dots, x_k$ for a suitably chosen $k<n$, and let $B$ be a solution for the remaining $n-k$ nails $x_{k+1},\dots, x_n$. Now take a commutator like $ABA^{-1}B^{-1}$ and substitute the $x_i$. Since the expressions we substitute for $A$ and $B$ are based on disjoint sets of generators, no cancellations can take place as it stands. But if any one of the generators $x_1,\dots,x_n$ is set to 1, one of $A$ and $B$ will collapse to the identity, and so will the commutator.
\end{proof}

By taking $k\approx n/2$, this divide-and-conquer approach gives a solution to the $n$ nails puzzle that corresponds to a word of length about $n^2$ in the fundamental group. Actually the length will be exactly $n^2$ whenever $n$ is a power of 2, with linear interpolation between the powers of 2. For instance, when $n=4,5,6,7,8$ we obtain words of length $16, 28, 40, 52, 64$. The question whether this is the minimal length of a nontrivial word that collapses on setting any one of the generators to 1 is mentioned in \cite{Demaine} and seems to be an open problem. A nontrivial lower bound was obtained in \cite{FulekAvvakumov}. 

As a side-remark, the case $n=3$ is related to another puzzle named \emph{Curve and Three Shadows} in \cite{Winkler}. The solutions found by John Terrell Rickard and Donald Knuth encode picture-hangings on three nails that will fall on the removal of any one of them, but in that puzzle, even more was required.  
 
Notice that the fundamental group consists of equivalence classes of curves under homotopy, meaning that the curve is allowed to pass through itself (just not through the nails). Therefore in order to get from a word like \begin{equation} \label{ex4} x_1x_2x_1^{-1}x_2^{-1} x_3x_4x_3^{-1}x_4^{-1}x_2x_1x_2^{-1}x_1^{-1}x_4x_3x_4^{-1}x_3^{-1}\end{equation} (for $n=4$) to a solution to the picture-hanging puzzle with a physical  wire that cannot pass through itself, we must be careful to make sure that the picture actually falls when a nail is removed. A construction by Michael Paterson depicted in \cite{Demaine} shows a more efficient solution relying on the fact that a physical wire can lock itself even though its path corresponds to the identity element of the fundamental group. We return to this issue in Section~\ref{S:physics}.

\section{Fundamental theorem of picture-hanging} \label{S:fundamental}

The group-theoretic setting allows us to investigate picture-hanging puzzles more generally. For instance, the word $xyzx^{-1}y^{-1}z^{-1}$ shows that we can hang a picture on three nails so that it stays on the removal of any one nail, but falls on the removal of any two. 

Considering other variations like the puzzles listed in \cite{Demaine}, we are led to the conjecture that all reasonable tasks of this kind are solvable, and this is indeed the case, as was shown in \cite{Demaine}. Here we present two slightly different proofs that are simpler than the one given in \cite{Demaine}.

In order to make precise the concept of a reasonable picture-hanging puzzle, we can encode a given task for $n$ nails as a boolean function of $n$ variables. Although it's not important, we prefer a convention opposite to the one in \cite{Demaine}, and take the $i$:th input variable to be \emph{true} if the $i$:th nail is present, and \emph{false} if it's removed. Consequently the output is true if the picture remains hanging, and -- homophonic mnemonic -- false if it falls. 

If $S$ is a set of generators/nails, we can denote by $\Lambda(S)$ a generic nontrivial word with the property (as in Theorem~\ref{T:obvious}) of collapsing to the identity whenever one of the generators in $S$ is set to 1. Here the letter $\Lambda$ is meant to resemble the symbol $\bigwedge$ that stands for logical AND.

Not all boolean functions can be realized as picture-hangings. For instance we can't ask the picture to hang if we remove all nails. It also seems that the function must be monotone: If the picture hangs on a certain set of nails, it can't fall because more nails are present. Topologically, a nail only restricts the movement of the wire and cannot allow it to untangle if it couldn't otherwise. And in the group theoretical setting, if a word simplifies to the identity, it can't evaluate to something else because more generators are set equal to 1. 
We return briefly in Section~\ref{S:physics} to the question whether monotonicity is actually necessary in a physical model.   

\begin{Thm} \label{T:fundamental}
Every nontrivial monotone boolean function can be realized as a picture-hanging.
\end{Thm}

\begin{proof} [First proof of Theorem~\ref{T:fundamental}]
For a given nontrivial monotone boolean function $f$, let $S_1, \dots, S_N$ be the \emph{minimal} sets where $f$ evaluates to \emph{true}, in other words the minimal sets of nails where the picture is supposed to  keep hanging. 
We claim that the word \begin{equation} \label{firstproof} \Lambda(S_1)\cdot \Lambda(S_2)\cdots \Lambda(S_N)\end{equation}
represents the function $f$.

In one direction it's clear: If we quotient out enough generators that none of the sets $S_i$ remains intact, then \eqref{firstproof} collapses to the identity, as it should. 

On the other hand it might not be obvious that \eqref{firstproof} can't also collapse if enough generators remain that several factors are nontrivial.
However, if the set of remaining generators/nails is \emph{equal} to one of the sets $S_i$, then by minimality \emph{only} the factor $\Lambda(S_i)$ is nontrivial, and the whole expression \eqref{firstproof} evaluates to that factor. Finally we already know that monotonicity holds: The expression \eqref{firstproof} can't become the identity because more than a minimal set of generators remain. 
\end{proof}

Our second proof is based on emulating logical AND- and OR-gates, as in the proof in \cite{Demaine}.
 
If $A$ and $B$ are words, the new word $ABA^{-1}B^{-1}$ might collapse even if both $A$ and $B$ are distinct from the identity: The equation $ABA^{-1}B^{-1} = 1$ is equivalent to $AB = BA$, which means that $ABA^{-1}B^{-1}$ collapses precisely when $A$ and $B$ commute. This can happen in some ways that aren't completely trivial, for instance if $A = xyx^{-1}$ and $B= xy^3x^{-1}$.
For this reason we can't always use the commutator $ABA^{-1}B^{-1}$ as an AND-gate: While it's true that the commutator becomes the identity unless both $A$ and $B$ are nontrivial, the converse is not true.

Similarly, $AB$ may collapse even if $A\neq 1$ and $B\neq 1$, simply by $A$ and $B$ being inverses. Therefore we cannot always use the product of two words as an OR-gate.  

Demaine et al \cite{Demaine} show how to construct ``safe'' logical gates by applying theorems of A.~I.~Mal'tsev and G.~A.~Gurevich. Here we point out that a simplification is possible. In order to obtain a safe OR-gate, we take two arbitrary generators $x$ and $y$ (we must assume that the number of generators/nails is at least 2), and \emph{pad} the words $A$ and $B$ by surrounding them by the symbols $x$, $y$ and their inverses in the following way:

\begin{equation} \label{safeOR} (x^MAx^{-M}) \cdot (y^MBy^{-M}).\end{equation}

Again it's clear that if $A = B =1$, then \eqref{safeOR} will collapse to 1 as well. Moreover, if one of $A$ and $B$ is nontrivial and the other is 1, then \eqref{safeOR} will be nontrivial. In the final case that $A$ and $B$ are both nontrivial, we claim that provided $M$ is large enough, \eqref{safeOR} can't collapse. This is based on the following easy lemma:

\begin{Lemma} 
If the symbols $x$ and $x^{-1}$ both occur fewer than $M$ times in the word $A\neq 1$, then \[ x^MAx^{-M}\] will simplify to a word that both begins and ends with one of the symbols $x$ or $x^{-1}$. 
\end{Lemma}  

\begin{proof}
If $A$ is a nontrivial power of $x$, then so is the simplified word. Otherwise there is some symbol other than $x$ and $x^{-1}$ that remains after simplifying $A$. In that case no symbol from the left padding will cancel against any symbol of the right padding. Since every cancellation will use up one symbol from $A$ and at most one from the padding, this means that at least one symbol will remain from each side of the padding.
\end{proof}

In the same way we can construct a safe AND-gate by starting from the expression $ABA^{-1}B^{-1}$ and pad to avoid unwanted cancellations:

\begin{equation} \label{safeAND} 
(x^MAx^{-M}) \cdot (y^MBy^{-M}) \cdot (x^MA^{-1}x^{-M}) \cdot (y^MB^{-1}y^{-M}).
\end{equation}
Again provided $x$ and $x^{-1}$ occur fewer than $M$ times in $A$, and $y$ and $y^{-1}$ occur fewer than $M$ times in $B$,  the expression \eqref{safeAND} can't collapse unless one of $A$ and $B$ does.
We thereby obtain a second proof: 
\begin{proof} [Second proof of Theorem~\ref{T:fundamental}] It's a fact of basic propositional logic that the nontrivial monotone boolean functions are precisely those that can be expressed by nesting AND- and OR-gates. An expression in terms of such gates can be translated to the required word in the free group by repeatedly applying \eqref{safeOR} and \eqref{safeAND}.  
\end{proof}

\section{Physics is complicated} \label{S:physics}
In order to argue that Theorem~\ref{T:fundamental} is the end of the story as far as solvability of picture-hanging puzzles goes, it seems we have to exclude some physics. With a wire that can't pass through itself and a finite gravitational field, we could solve even some non-monotone puzzles. 

\begin{figure} [h]
\begin{center}
\begin{tikzpicture} [scale=0.8]

\draw[ultra thick]  (-0.1,3.1) -- (0.4, 4.1);

\draw[ultra thick]  (-0.5,0.5) -- (-0.5, 2) -- (1,2) -- (1,0.5) -- cycle;

\draw[blue, ultra thick] (0,2) .. controls (0.5, 4) and (0, 4) .. (-1.2,2);
\draw[blue, ultra thick] (-1.2, 2) .. controls (-2.7, -0.5) .. (0,-0.5);

\draw[blue, ultra thick] (0,-0.5) .. controls (3, -0.5) and (3,0) .. (0,0);

\draw[blue, ultra thick] (0, 0) .. controls (-2.4, 0) .. (-1.2,1.9);
\draw[blue, ultra thick] (-1.2, 1.9) .. controls (0, 3.8) .. (0.5,2);

\filldraw (-0.1,3.1) circle (0.2); 
\filldraw (-2.2,-0.2) circle (0.2);
\draw[ultra thick]  (-2.2,-0.2) -- (-1.7, 0.8);

\node at (-0.4, 3.8) {\Large $A$};
\node at (-2.6, 0.4) {\Large $B$};

\begin{scope}[xshift = 7cm]

\draw[ultra thick]  (-0.1,3.1) -- (0.4, 4.1);

\draw[ultra thick]  (-0.5,0.5) -- (-0.5, 2) -- (1,2) -- (1,0.5) -- cycle;

\draw[blue, ultra thick] (0,2) .. controls (0.5, 4) and (0, 4) .. (-1.2,2);
\draw[blue, ultra thick] (-1.2, 2) .. controls (-2.7, -0.5) .. (0,-0.5);

\draw[blue, ultra thick] (0,-0.5) .. controls (3, -0.5) and (3,0) .. (0,0);

\draw[blue, ultra thick] (0, 0) .. controls (-2.4, 0) .. (-1.2,1.9);
\draw[blue, ultra thick] (-1.2, 1.9) .. controls (0, 3.8) .. (0.5,2);

\filldraw (-0.1,3.1) circle (0.2); 
\filldraw (-2.2,-0.2) circle (0.2);
\draw[ultra thick]  (-2.2,-0.2) -- (-1.7, 0.8);

\node at (-0.4, 3.8) {\Large $A$};
\node at (-2.6, 0.4) {\Large $B$};

\filldraw (1.8,-0.2) circle (0.2);
\draw[ultra thick]  (1.8,-0.2) -- (2.3, 0.8);
\node at (2.6, 0.4) {\Large $C$};

\end{scope}
\end{tikzpicture}
\caption{Left: A counterexample to monotonicity. If we hold the picture and wire in this position and let go, the nail $A$ will slow the picture down enough that everything falls. But if the nail $A$ hadn't been there, the picture would have fallen straight through the loop and become stuck around $B$. 
Right: The condition for the picture to remain hanging might not even be a function of the set of remaining nails. If we remove nail $A$, the picture will fall through the loop below it. If we then remove nail $C$, the picture will still be attached around nail $B$. If on the other hand we first remove $C$, the picture already falls (and remains fallen if we then remove nail $A$).}
\label{F:counterMonotonicity}
\end{center}
\end{figure}
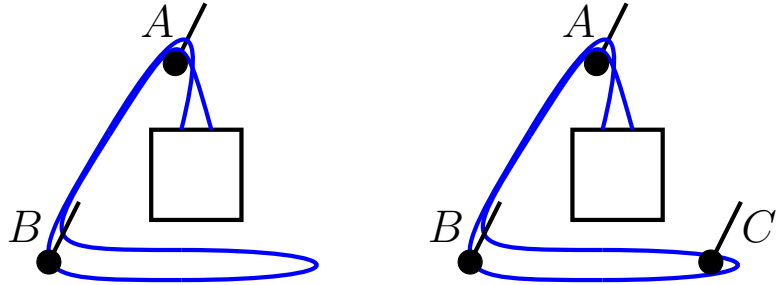

Figure~\ref{F:counterMonotonicity} (left) shows a situation where the picture would fall if we let go, while it would have remained hanging if the nail $A$ hadn't been there. This example might seem contrived, as we normally think of the picture as hanging in equilibrium and the nails being removed one by one. 

But then it's even worse: The condition for the picture to remain hanging might not even be a function of the set of remaining nails, but could depend on the order in which they are removed. To see this, we just modify the previous example using a third nail $C$ as in Figure~\ref{F:counterMonotonicity} (right).

These examples obviously depend on a number of assumptions about the exact positioning of the nails and the wire, friction, gravity, and so on. Here we simply conclude that this is too complicated, and in the following we consider the group-theoretical model only.

That said, it would obviously be interesting to have an improvement of the results of the following section in the spirit of the Paterson solution \cite{Demaine} to the original $n$-nails puzzle.

\section{The $k$-out-of-$n$ puzzle}

As was pointed out by Demaine et al \cite{Demaine}, the sheer number of monotone boolean functions on $n$ variables (the so-called Dedekind number) implies that most of them will require almost $2^n$ symbols for a representation as a word in the free group. In this perspective, the so-called $k$-out-of-$n$ problem, where the picture should remain hanging if and only if at least $k$ out of $n$ nails remain, turns out to admit an unusually efficient representation. 

In \cite{Demaine}, a construction based on a sorting network shows that $k$-out-of-$n$ has a representation of length $O(n^{1,561,600})$. The exponent can be improved by simplifying the logical gates as in Section~\ref{F:fundamental}, but it will still be in the thousands as long as it relies on sorting networks of the AKS-type. Here we establish a polynomial bound with a reasonable exponent, although we don't have an efficient method of \emph{computing} what that representation looks like.

\subsection*{Divide-and-conquer}
We first mention a divide-and-conquer solution that achieves a representation of ``quasi-polynomial''  length $n^{O(\log n)}$. This is already shorter than ``most'' monotone boolean functions. Conceivably the method can be improved, but it doesn't quite seem to give a polynomial bound.

Suppose we have $2n$ generators $x_1,\dots, x_{2n}$. Inductively let $A_k$ be a representation of \emph{at least $k$ out of $x_1,\dots,x_n$}, and similarly let $B_k$ represent \emph{at least $k$ out of $x_{n+1},\dots,x_{2n}$} (if $k>n$, then $A_k = B_k = 1$). Now we can represent $k$ out of $2n$ by 
\begin{equation} \label{divideandconquer} B_k\cdot (A_1B_{k-1}A_1^{-1}B_{k-1}^{-1}) \cdot (A_2B_{k-2}A_2^{-1}B_{k-2}^{-1}) \cdots (A_{k-1}B_1A_{k-1}^{-1}B_1^{-1}) \cdot A_k.\end{equation}
If fewer than $k$ generators remain, \eqref{divideandconquer} will collapse, while if exactly $k$ remain, then exactly one of the factors will be nontrivial. We don't need any padding to make each factor safe, since $A_i$ and $B_{k-i}$ involve disjoint sets of generators. 

An inductive argument shows that if $n$ is a power of 2, say $n=2^a$, then the length of the word representing $k$-out-of-$n$ will be at most
\[ 2^{a(a+3)/2} = n^{(a+3)/2} = n^{\frac12\log_2{n} + 3/2}.\]
This bound can be improved slightly, but we don't seem to get a polynomial bound on the length of the word in this way. 

\subsection*{A probabilistic approach}
Using instead a probabilistic method, we show that for arbitrary $k$ and $n$ with $1\leq k\leq n$, there exists a word of length $O(n^{7.004})$ in the generators $x_1,\dots, x_n$ that remains nontrivial if all but $k$ generators are set to 1, but collapses if only $k-1$ generators remain. The proof is based on a method of Leslie Valiant \cite{Valiant}, see also \cite{Damgaard, Goldreich}.
 
\begin{Thm} \label{T:main}
For $n$ and $k$ with $1\leq k\leq n$, the $k$-out-of-$n$ function can be represented as a picture-hanging by a word of length $O(n^c)$, where 
\begin{equation} \label{defc} c = \log_{3/2}(6) + \log_2(6) \approx 7.004.\end{equation}
\end{Thm}

\subsection*{The case $n=2k-1$}
We first establish the special case of $n= 2k-1$, where $n$ is odd and the picture is required to hang precisely when a majority of the nails remain. Then we show how to modify the argument for the case of general $k$.

We construct recursively a sequence $W_0, W_1, W_2,\dots$ of random words in the generators $x_1,\dots, x_n$ that will eventually tend to have high probability of evaluating to something nontrivial if $k$ of the generators remain, and to the identity if only $k-1$ of them remain.

For convenience, in the following analysis of probabilities we can think of the two specific scenarios of retaining the first $k-1$ and the first $k$ nails respectively. Therefore we let $\phi_i$ be the homomorphism of the free group generated by $x_1,\dots, x_n$ to itself given by mapping $x_1,\dots, x_i$ to themselves, and all the following generators to 1. Since the following construction of $W_0, W_1, W_2,\dots$ will be symmetric under all permutations of the generators, it suffices to analyze the probabilities that $\phi_{k-1}(W_d) =1$ and that $\phi_k(W_d) =1$ respectively.

The depth zero word $W_0$ is just a single symbol uniformly chosen from the $n$ symbols $x_1,\dots, x_n$.
Clearly the probability that $W_0$ ``hangs'' on the first $i$ nails is given by \[ P(\phi_i(W_0)\neq 1) = \frac in.\]

We would like it to hang on the first $k$ nails but not on the first $k-1$, so we can say that the probability $p_0$ of ``failure'' is the same in both directions:

\[ p_0 = P(\phi_{k-1}(W_0)\neq 1) = P(\phi_k(W_0) = 1) = \frac12 - \frac1{2n}.\] 

To build recursively the random depth $d+1$ word $W_{d+1}$, we combine three independently chosen words of depth $d$ with a \emph{safe majority-gate}. The safe majority gate is constructed from the expression $ABCA^{-1}B^{-1}C^{-1}$ by padding:

\begin{equation} \label{safeMAJORITY} 
x^MAx^{-M}y^MBy^{-M}z^MCz^{-M}x^MA^{-1}x^{-M}y^MB^{-1}y^{-M}z^MC^{-1}z^{-M}.
\end{equation}

As the depth $d$ increases, the word $W_d$ will distinguish more and more clearly between the cases of a majority or minority of generators remaining. We wish to bound the probability of failure, and therefore let 
\[ p_d = P(\phi_{k-1}(W_d)\neq 1) = P(\phi_k(W_d) = 1).\] 
Since $W_{d+1}$ fails precisely when at least two out of three independent words of distribution $W_d$ fail, we have 
\[p_{d+1} =p_d^3 + 3p_d^2(1-p_d) = 3p_d^2 -2p_d^3.\]

Next we want to show that $p_d$ approaches zero reasonably quickly. We analyze separately the first phase, where $p_d$ is close to $1/2$ but its distance to $1/2$ grows exponentially, and the second phase where $p_d$ is close to zero and is essentially squared in each step.
 
When $p_d$ is close to $1/2$, the distance to $1/2$ will increase by roughly a factor $3/2$ for each new level:
\begin{Lemma}
\[ \frac{1/2 - p_{d+1}}{p_{d+1}} \geq \frac32\cdot  \frac{1/2 - p_d}{p_d}.\]
\end{Lemma}
\begin{proof}
We let $p_d=p$ and plug in $p_{d+1} = 3p^2-2p^3$. Then we are comparing 
\[  \frac{1/2 - 3p^2 + 2p^3}{3p^2 - 2p^3} \quad \text{ to } \quad \frac32\cdot\frac{1/2 - p}{p}. \]
After clearing the denominators, we find that the difference 
\[ 2p(1/2 - 3p^2 + 2p^3) - 3(1/2-p)(3p^2 - 2p^3)\] factorizes as 
\[ 2p(2-p)(p-1/2)^2,\] which is nonnegative in the interval $0\leq p\leq 1/2$.
\end{proof}

It follows inductively that 
\[ \frac{1/2 - p_d}{p_d} \geq \frac{1/2 - p_0}{p_0} \cdot \left(\frac32\right)^d \geq n\cdot \left(\frac32\right)^d.\] 
This means in particular that if $p_d\geq 1/4$, then $(3/2)^d \leq n$. We conclude that in order to ensure that $p_d\leq 1/4$, it suffices to go to depth $d\geq \log_{3/2}(n)$.

After just one more step, the probability of failure will be smaller than $1/6$, and will then quickly approach zero, roughly squaring in each step:

\begin{Lemma}
\[ 3p_{d+1} \leq (3p_d)^2.\]
\end{Lemma}

\begin{proof}
\[ 3\cdot (3p^2-2p^3) = 9p^2-6p^3 \leq 9p^2 = (3p)^2.\]
\end{proof}
Therefore, if $p_a \leq 1/6$ so that $3p_a\leq 1/2$, then \[3p_{a+b} \leq (1/2)^{2^b}.\]
It follows that \[p_d < 2^{-n}\] for some 
\begin{equation} \label{suchd} d = \log_{3/2}(n) + \log_2(n) + O(1).\end{equation}

Since the construction of $W_d$ is symmetric in the generators $x_1,\dots, x_n$, the estimated failure probability holds for every possible state (present or removed) of the set of nails.
We conclude that for $d$ as in \eqref{suchd}, the probability that $W_d$ fails on as much as a single one of the $2^n$ states of the nails is smaller than 1. Therefore there has to be a word among the possibilities for $W_d$ that doesn't fail on any of the states.

Next we analyse the length of the word $W_d$. If it wasn't for the padding, $W_d$ would have length 
\[ 6^d  = O\left(6^{\log_{3/2}(n) + \log_2(n)}\right) = O\left(n^{ \log_{3/2}(6) + \log_2(6)}\right)\]
in accordance with \eqref{defc}.

It can be verified that the padding in total only increases the length by another constant factor: Before the words reach length $2n$, the padding consists of only one symbol at each end. In this phase, the length (1, 18, 120 etc) of the word $X_d$ is exactly \[ \frac{17\cdot 6^d - 12}5,\] and in particular only imposes an extra factor bounded by $17/5$.
Once the words $W_d$ reach length $2n$ (so that they might have one pair of inverses of each of the $n$ generators), the contribution of the padding can be estimated by a factor $1+2/n$ in each step, which grows to a factor of $1+O(\log n /n)$ over the remaining $O(\log n)$ steps. 

Notice here that we can choose the padding after the random choices of generators in the words.

\subsection*{General $k$}
Finally we establish Theorem~\ref{T:main} for general $k$. 
To do this, we modify the first step of the construction by letting $W_0$ be a product of a suitably chosen, possibly random, number $m$ of symbols taken uniformly from $x_1,\dots, x_n$ and conditioned on being distinct. That is, given $m$, they are chosen with uniform distribution on the $n \choose m$ sets of $m$ generators.
This product will act as an OR-gate, with nontrivial evaluation if and only if at least one of the generators in the word is not quotiented out. 

We still want the ``failure probability'' to be the same in both directions, that is, we want to choose $m$ so that
\[ P(\phi_{k-1}(W_0)\neq 1) = P(\phi_k(W_0)=1).\]
This is not possible for a fixed value of $m$ except in some special cases (like $k=3$, $n=9$, $m=2$). But it can obviously be achieved if we choose $m$ randomly, and we can let $m$ have support on the two consecutive integers where $P(\phi_{k-1}(W_0)\neq 1)$ overtakes $P(\phi_k(W_0)=1)$.
If $k>n/2$, $m$ will be either 0 or 1 and the following analysis rather trivial, but the case $k<n/2$ requires a little bit of calculation.

The failure probabilities can be thought of in terms of the minimal index in the word $W_0$: Relative to $\phi_{k-1}$, the word $W_0$ fails when the minimal index is strictly smaller than $k$ so that $\phi_{k-1}(W_0) \neq 1$. On the other hand relative to $\phi_k$, it fails if the minimal index is strictly larger than $k$ so that $\phi_k(W_0)=1$.

Assuming that the two failure probabilities are equal, we get a lower bound on the probability that the minimum index is equal to $k$:
\begin{multline} \notag P(\text{minimum index $=k$}) \\= P(\text{$x_k$ occurs in $W_0$}) \cdot P(\text{remaining indices $>k$}) 
 \\ \geq \frac{E(m)}n \cdot \frac12, \end{multline}
since the probability that all indices are $\geq k$ is at least $1/2$ even without conditioning on one of them being equal to $k$.
This gives a bound away from $1/2$ on the failure probability $p_0$:
 \begin{equation} \label{boundaway} p_0 = \frac{1 - P(\text{minimum index $=k$})}{2} \leq \frac12 - \frac{E(m)}{4n}.\end{equation}
 
 Even for large $k$ we have $E(m)\geq 1/2$. Therefore in the regime $n/2 \leq k \leq n$ it's clear that the modified first step only affects the implied constant in Theorem~\ref{T:main}. 

For small $k$, and thereby possibly large $m$, the modified first step imposes an extra factor of $E(m) + O(1)$ on the length of the word $W_0$ and thereby on all subsequent words. But by \eqref{boundaway}, it also gives an extra factor of order $m$ in $1/2 - p_0$. This is better than what is provided by the initial phase of the majority function, where increasing $1/2 - p_d$ by a factor $3/2$ imposes a factor $6$ on the length of the word, and therefore increasing it by a factor $t$ increases the length of the word by $t^{\log_{3/2}(6)}$. 

Therefore when $m$ is large, the extra length of the word $W_0$ will be more than compensated for by fewer rounds of the majority function needed. The case of general $k$ therefore gives, up to a constant factor, the same bound on the length of the word as the case $n=2k-1$. This concludes the proof of Theorem~\ref{T:main}.

\subsection*{Finding hay in a haystack}
While rigorously demonstrating the existence of a polynomial-size word that solves the $k$-out-of-$n$-problem, our argument doesn't provide an efficient recipe for finding it. 

But the situation isn't that bad: if we take the recursion involving the majority function just one more step than is needed in the proof of Theorem~\ref{T:main}, the failure probability on random input will essentially drop to just $2^{-2n}$. This means that the probability that there is even a single state where the word fails is now exponentially small. The computational problem therefore has the character of ``finding hay in a haystack''.

\end{document}